\documentclass{article}
\usepackage{style}

\title{Hardness of clique approximation\\ for monotone circuits}
\author{Jarosław Błasiok}
\author{Linus Meierhöfer}
\affil{ETH Zürich}

\begin{document}

\maketitle
\begin{abstract}
    We consider a problem of approximating the size of the largest clique in a graph, with a monotone circuit. Concretely, we focus on distinguishing a random \ErdosRenyi graph $\mathcal{G}_{n,p}$, with $p=n^{-\frac{2}{\alpha-1}}$ chosen st. with high probability it does not even have an $\alpha$-clique, from a random clique on $\beta$ vertices (where $\alpha \leq \beta$). Using the approximation method of Razborov, Alon and Boppana showed in their influential work in 1987 that as long as $\sqrt{\alpha} \beta < n^{1-\delta}/\log n$, this problem requires a monotone circuit of size $n^{\Omega(\delta\sqrt{\alpha})}$, implying a lower bound of $2^{\tilde\Omega(n^{1/3})}$ for the exact version of the $\cliquep_k$ problem when $k\approx n^{2/3}$.  Recently Cavalar, Kumar, and Rossman improved their result by showing the tight lower bound $n^{\Omega(k)}$, in a limited range $k \leq n^{1/3}$, implying a comparable $2^{\tilde{\Omega}(n^{1/3})}$ lower bound after choosing the largest admissible~$k$.

    We combine the ideas of Cavalar, Kumar and Rossman with the recent breakthrough results on the sunflower conjecture by Alweiss, Lovett, Wu and Zhang to show that as long as $\alpha \beta < n^{1-\delta}/\log n$, any monotone circuit rejecting $\mathcal{G}_{n,p}$ graph while accepting a $\beta$-clique needs to have size at least $n^{\Omega(\delta^2 \alpha)}$; this implies a stronger $2^{\tilde{\Omega}(\sqrt{n})}$ lower bound for the unrestricted version of the problem.

    We complement this result with a construction of an explicit monotone circuit of size $O(n^{\delta^2 \alpha/2})$ which rejects $\mathcal{G}_{n,p}$, and accepts any graph containing $\beta$-clique whenever $\beta > n^{1-\delta}$. In particular, those two theorems give a precise characterization for the largest $\beta$-clique that can be distinguished from $\mathcal{G}_{n, 1/2}$:
    when $\beta > n / 2^{C \sqrt{\log n}}$, there is a polynomial size circuit that solves it, while for $\beta < n / 2^{\omega(\sqrt{\log n})}$ every circuit needs size~$n^{\omega(1)}$.
\end{abstract}

\section{Introduction}
Monotone circuits form a restricted computation model, where computation is performed by a directed acyclic graph, with input nodes (of degree $0$) labeled by the input variables and internal nodes (\emph{gates}) labeled $\land$ or $\lor$ (computing logical \emph{and}, or logical \emph{or} respectively, of input wires). It is not difficult to see that every monotone function on $n$ binary variables $f:\{0, 1\}^n \to \{0, 1\}$ (and only monotone functions), can be computed by such a circuit, and by a simple counting argument one can show that a random monotone function requires a monotone circuit of size $2^{\Omega(n)}$. 

Remarkably, in contrast with more general boolean circuits (where the negation gate $\lnot$ is also allowed), since the breakthrough work of Razborov~\cite{razborov1985lower}, there are known unconditional super-linear lower bounds for the size of monotone circuit computing explicit monotone functions. Concretely, Razborov showed that for any $k \leq \log(n)$, the function $\cliquep_k : \{0,1\}^{\binom{n}{2}} \to \{0,1\}$~--- interpreting its input as an adjecency matrix of a graph $G$, and outputting $1$ if and only if said graph contains a clique on $k$ vertices~--- requires a monotone circuit of size $n^{\Omega(k)}$. Putting $k = \log n$, this gives a quasipolynomial lower bound $2^{\Omega(\log^2(n))}$;~an analogous result for boolean circuits including negation would imply $P \not=NP$ and seems to be way out of reach for current techniques, almost 40 years later.

To prove his lower bound, Razborov used the sunflower lemma of~\Erdos and Rado~\cite{erdos60}, bringing it to the attention of the theoretical computer science community. In the highly influential follow-up work, Alon and Boppana~\cite{alon87} further utilized the approximation method introduced by Razborov. They showed a $n^{\Omega(\sqrt{k})}$ lower bound for the $\cliquep_k$ problem when $k \leq n^{2/3}/\log n$~--- implying a $2^{\tilde{\Omega}(n^{1/3})}$ lower bound at the optimal value of $k=n^{2/3} / \log n$~--- a result which has since become a landmark of circuit complexity. Interestingly, and less commonly known, in the same paper, they showed a stronger inapproximability result: any monotone circuit that rejects every graph without even an $\alpha$-clique and accepts every graph that has a $\beta$-clique (where $\alpha \leq \beta$) needs to have a size at least $n^{\Omega(\delta \sqrt{\alpha})}$ as long as $\sqrt{\alpha} \beta < n^{1-\delta}/\log n$. In fact, their result holds in the average case -- they show that it is hard to reject a random $(\alpha-1)$-partite graph while accepting a clique on $\beta$ random vertices.

Since then, several other techniques for showing monotone circuit lower bounds for specific problems have been introduced over almost four decades~(e.g. \cite{andreev1985method,andreev1987method,jukna1999combinatorics,garg2018monotone,goos2019adventures,cavalar2023}), but the result of Alon and Boppana stood as essentially the best-known lower bound for the clique problem.

Very recently, the breakthrough result of Alweiss, Lovett, Wu, and Zhang~\cite{alweiss20} made a major step towards resolving the sunflower conjecture~--- conjecture about the ``right'' quantitative dependency in the sunflower lemma --- and was followed in short succession by a sequence of further improvements \cite{bell2021note,Rao2020Coding,Tao2020sunflower, Hu2021sunflower, rao2023sunflowers}. Particularly noteworthy is the exposition by Rao \cite{rao2023sunflowers}, presenting not only the improved sunflower bounds but also the resolution of the Kahn-Kalai conjecture in just a few pages.

Following up on these results, Cavalar, Kumar, and Rossman~\cite{cavalar2022}, showed how to utilize the new sunflower bounds to prove the first $2^{\Omega(\sqrt{n})}$ lower bound for any explicit monotone function (improving the earlier $2^{\tilde{\Omega}(n^{1/3})}$ lower bound), specifically for the~Harnik-Raz function\cite{harnik00lower}. In the same paper, they also gave the first substantial improvement over the Alon-Boppana bound for clique: they showed that as long as $k \leq n^{1/3 - \delta}$ the $\cliquep_k$ problem requires $n^{\Omega(k)}$ circuit size~--- a tight lower bound, but in a further restricted range of parameters; choosing the optimal $k \approx n^{1/3}$, this leads to the same $2^{\tilde{\Omega}(n^{1/3})}$ lower bound for $\cliquep$. As it turns out, their result for the clique did not use the new sunflower bounds, and in fact, it was unclear how to combine those techniques.

\subsection{Our results}


We prove a lower bound on the size of monotone circuits solving the promise problem $\mathrm{GAP}_{\alpha, \beta}(n)$.

\begin{definition}
    For integers $n, \alpha, \beta\in \mathbb{N}$ with $ \alpha < \beta$, define the promise problem $\mathrm{GAP}_{\alpha, \beta}(n)$, where the two promise subsets $\mathcal{A}, \mathcal{B}$ are defined as
    $$\mathcal{A} := \{G \in \{0, 1\}^{\binom{n}{2}}\,|\, G \text{ does not contain an } \alpha \text{-clique}\},$$
    and
    $$\mathcal{B} := \{G \in \{0, 1\}^{\binom{n}{2}} \,|\, G \text{ contains an } \beta \text{-clique}\}.$$
\end{definition}


To prove a lower bound on $\mathrm{GAP}_{\alpha, \beta}(n)$, we introduce an associated strengthened distributional problem $\mathcal{D}(\mathcal{T}^-_\alpha, \mathcal{T}^+_\beta)$ to distinguish two graph distributions $\mathcal{T}^-_\alpha$ and $\mathcal{T}^+_\beta$.

\begin{definition}
    For an integer $n, \alpha, \beta$ with $\alpha \leq \beta$, let $\mathcal{T}^+_\beta$ be a uniformly random distribution of isolated $\beta$-cliques on $n$ vertices. Let $\mathcal{T}^-_{\alpha} \sim \mathcal{G}_{n, p}$ the \ErdosRenyi random graph distribution on $n$ vertices with probability parameter $p=n^{-2/(\alpha - 1)}$. 

    The problem $\mathcal{D}(\mathcal{T}^-_\alpha, \mathcal{T}^+_\beta)$ is to distinguish those two distributions~--- i.e. given on input a graph $G$ drawn either from $\mathcal{T}^-_{\alpha}$ or $\mathcal{T}^+_{\beta}$ we want to reject with probability at least $2/3$ in the former case and accept with probability at least $2/3$ in the latter.
\end{definition}
Note that $\alpha$ is chosen so that $\mathcal{D}(\mathcal{T}^-_\alpha, \mathcal{T}^+_\beta)$ is associated with $\mathrm{GAP}_{\alpha, \beta}(n)$ as the probability that a graph drawn from $\mathcal{T}^-_\alpha$ contains a clique of size $\alpha$ is small~--- a simple union bound can be used to argue that $\Pr_{G\in \mathcal{T}^-_\alpha}[G \text{ contains } \alpha \text{-clique}] < 1/4$ for $\alpha \geq 4$. As such, any circuit solving the problem $\mathrm{GAP}_{\alpha, \beta}$ can distinguish $\mathcal{T}^-_{\alpha}$ from $\mathcal{T}^+_{\beta}$, and conversely, a lower bound for $\mathcal{D}(\mathcal{T}^-_{\alpha}, \mathcal{T}^+_{\beta})$ implies a lower bound for $\mathrm{GAP}_{\alpha, \beta}$. 

Our main results are lower and upper bounds on the monotone complexity of $\mathcal{D}(\mathcal{T}^-_\alpha, \mathcal{T}^+_\beta)$, captured in the following two theorems.


\begin{theorem}
\label{thm:main-lb}
    Let $\alpha\beta < n^{1-\delta}/\log(n)$. Then there exists no monotone circuit solving $\mathcal{D}(\mathcal{T}^-_\alpha, \mathcal{T}^+_\beta)$ with size less than $n^{\Omega(\delta^2 \alpha)}$.
\end{theorem}

\begin{theorem}
\label{thm:main-ub}
    Let  $\beta \geq n^{1-\delta}$. Then there exists a monotone circuit $C$ on $\binom{n}{2}$ inputs, solving $\mathcal{D}(\mathcal{T}^-_\alpha, \mathcal{T}^+_\beta)$ with $size(C) = \mathcal{O}(n^{\delta^2\alpha/2})$.
\end{theorem}

We would like to emphasize a special case of $\alpha = \log n$ (or, equivalently, the negative distribution $\mathcal{T}^-$ being an \ErdosRenyi graph $\mathcal{G}_{n,p}$ with $p=1/2$). In this case, according to \Cref{thm:main-ub}, we can choose $\delta = \sqrt{1/\log n}$ to get a polynomial-size circuit, rejecting $\mathcal{G}_{n,p}$ yet accepting a clique of size $\beta = n / 2^{C \sqrt{\log n}}$. On the other hand, \Cref{thm:main-lb} states that if we wanted to reject $\mathcal{G}_{n, 1/2}$ and accept a clique of size $\beta$, where $\beta = n / 2^{\omega(\sqrt{\log n})}$, we would need a circuit of size $n^{\omega(1)}$. 

As such, our two theorems provide a near tight characterization of a power of polynomial-size monotone circuits to distinguish the $\mathcal{G}_{n, 1/2}$ graph from a large clique. 

Another interesting corollary of~\Cref{thm:main-lb}, is obtained by setting $\beta = \alpha$, and taking $\alpha$ as large as possible while satisfying the restriction $\alpha^2 < n^{1-\delta}/\log n$, in order to obtain a strongest possible lower bound for the clique problem.

\begin{corollary}
    For any $k \leq n^{1/2 - \delta} / \log^2 n$, the monotone complexity of the $\cliquep_k$ is $n^{\Omega(\delta^2 k)}$. In particular monotone complexity of the $\cliquep$ problem is $2^{\tilde{\Omega}(\sqrt{n})}$.
\end{corollary}

Note that no truly exponential lower bounds for the size of the monotone circuit are known for \emph{any} explicit monotone problem. A lower bound of $2^{\tilde\Omega(\sqrt{n})}$  (where $n$ is size of the input) was shown for the arguably less-natural Harnik-Raz function \cite{harnik00lower} in \cite{cavalar2022}, breaking the long-standing barrier of $2^{\tilde{\Omega}(n^{1/3})}$ --- and using the same breakthroughs in sunflower lemmas that are crucial in our improvement.

In our case, the input is an adjecency matrix of a graph on $n$ vertices, hence the input has $m = \Theta(n^2)$ bits~--- our lower bound, in terms of the input size, is then $2^{\Omega(m^{1/4})}$~--- so the exponent is still by a factor of two worse than in the strongest known lower bound for any explicit monotone function.
\subsection{Approximation method \label{sec:intro-approx-method}}
Let us briefly recall the idea behind the approximation method for showing monotone circuit lower bounds, introduced by Razborov~\cite{razborov1985lower}, and then utilized by~\cite{alon87,cavalar2022}.

In order to show a lower bound for a size of a monotone circuit distinguishing two specific distributions $\mathcal{T}^+$ and $\mathcal{T}^-$, we proceed by inductively (gate-by-gate) approximating any given circuit $C$, by a simpler circuit $\hat{C}$ (from some class of simple circuits).

If we can show that
\begin{enumerate}
    \item Every simple circuit fails to distinguish between $\mathcal{T}^+$ and $\mathcal{T}^-$~--- that is, for all simple $\hat{C}$, we have $\E_{x \sim \mathcal{T}^+}\hat{C}(x) - \E_{x\sim \mathcal{T}^-}\hat{C}(x) \leq o(1)$.
    \item In each gate by applying the approximation we introduce error at most $\delta$ on both positive and negative distribution.
\end{enumerate}
That implies a lower bound $\Omega(\delta^{-1})$ for the size of the smallest monotone circuit $C$ distinguishing those two distributions.

For the clique problem (say, distinguishing random $\mathcal{G}_{n,p}$ graph which likely does not even have a clique of size $\alpha$, from a uniformly random clique of size $\beta \geq \alpha$), following \cite{razborov1985lower,alon87,cavalar2022}, the ``simple'' circuits we use to approximate a given circuit are just small DNF formulas, specifically conjunctions of clique indicators~--- that is circuits of form
\begin{equation*}
    \bigvee_i \mathcal{K}_{A_i},
\end{equation*}
where
\begin{equation*}
    \mathcal{K}_{A_i} := \bigwedge_{ \{u, v\} \in A_i} x_{uv}.
\end{equation*}
Moreover, a circuit is simple if the size of each clique indicator is bounded by $c$ (we will eventually chose $c := \delta \alpha$), and the number of clique indicators of each size $\ell \leq c$ is appropriately bounded.

At a given gate, we wish to approximate a conjunction or a disjunction of two such simple circuits again by a simple circuit. By applying de Morgan laws in the conjunction case, we can transform the circuit again into DNF (without introducing any error).

Then we proceed by repeating the following three steps, transforming the DNF obtained this way into a simple one.
\begin{enumerate}
    \item We replace all conjunctions $\mathcal{K}_{A_i} \land \mathcal{K}_{A_j}$ by clique indicators $\mathcal{K}_{A_i \cup A_j}$.
    \item As long as there is a family of cliques on sets $A_{i_1}, \ldots A_{i_k}$ having a specific combinatorial structure, we replace all of them by a clique on the set $C$, where $C$ is the intersection of $A_{i_j}$.
    \item We remove all clique indicators $\mathcal{K}_{A_i}$ for $A_i$ larger than threshold $c$.
\end{enumerate}

As it turns out, both on the positive and negative distributions, the step 1 does not introduce any error.

Step 3 can introduce error only on the positive distribution, and this error can be bounded by union bound --- any given large clique indicator is unlikely to be satisfied by a large random clique, and the number of those large clique indicators we discard is bounded (otherwise, we would have been able to apply step 2).

Finally, step 2 is the crux of the argument --- we need to show that if the number of indicators of a given size is too large, then there always is a subset of those indicators that can be replaced by its intersection $C$ (the \emph{core}), without introducing too much error on the negative distribution (clearly we will not introduce any error on the positive distribution in this case). 

In the Razborovs proof, and some presentations of the Alon-Boppana proof (see for example~\cite[Chapter 9]{jukna2012boolean}), one can focus on finding a sunflower consisting of many sets among the set system $\{A_i\}$ (a sunflower is a family of sets for which all pairwise intersections are the same, see \Cref{sec:intro-sunflowers}), and show that replacing a large sunflower by its core introduces only small error on the negative distribution. This, together with the \ErdosRado result that any large enough family of bounded sets contains a large sunflower, gives an upper bound on the number of clique indicators of size $\ell$ for a ``simple circuit'' (one on which the step 2 can no longer be applied), and hence yields a complete proof of a monotone circuit lower bound for the clique problem. Specifically, going carefully through the calculations, one could show this way a lower bound $n^{\Omega(\sqrt{k})}$ for the $\cliquep_k$ problem whenever $k \leq n^{1/2 - \delta}$, translating to a $2^{\tilde{\Omega}(n^{1/4})}$ lower bound for the $\cliquep$ problem after taking optimal $k$.
\footnote{In the actual Alon-Boppana paper, they used a slightly more general combinatorial structure: a sequence of distinct sets $A_{i_1}, \ldots A_{i_k}$, together with a set $C \subset A_t$ (not necessarily distinct from $A_{i_1}, \ldots A_{i_k}$), s.t. all pairwise intersections $A_{i_s} \cap A_{i_r} \subset C$~--- and they replaced $A_{i_1}, \ldots A_{i_k}$ by $C$ in step 2 here. Clearly any sunflower $A_{i_1} \ldots A_{i_k}$ with the core $C = \bigcap A_{i_j}$ satifies this property. By using this more general structure, they were able to show a lower bound $n^{\sqrt{k}}$ for $k\leq n^{2/3}$~--- matching the result one would get insisting on using a sunflower here if the sunflower conjecture was true; yet without having to prove this conjecture. }

The work of Cavalar, Kumar and Rossman \cite{cavalar2022}, introduced a notion of \emph{robust clique-sunflower}, which abstracts exactly the property needed to bound the error introduced on the negative distribution $\mathcal{G}_{n,p}$  by replacing the robust clique sunflower by its core. They showed better quantitative bounds on the size of the set system needed to contain such a robust clique sunflower, and were able to deduce a lower bound $n^{\tilde{\Omega}(k)}$ for $k \leq n^{1/3 - \delta}$ --- matching the same $2^{\tilde\Omega(n^{1/3})}$ for the clique problem when picking largest admissible $k$.

\subsection{Sunflowers, robust sunflowers and robust clique sunflowers\label{sec:intro-sunflowers}}

We define the robust clique sunflower (after~\cite{cavalar2022})~--- a notion abstracting the main property of combinatorial sunflowers used to obtain a monotone lower bound for the $\cliquep_k$ problem, where the negative distribution is $\mathcal{G}_{n,p}$. This property ensures that for a DNF formula 
\begin{equation*}
    \bigvee_{i} \mathcal{K}_{A_i}
\end{equation*}
if a number of those sets $A_i$ form a robust clique sunflower, we can replace all of the sets from the sunflower by a single clique indicator on the common intersection $C$, while introducing only a small error on the negative distribution (and introducing no error on the positive distribution).

\begin{definition}[Robust clique sunflower]
\label{def:robust-clique-sunflower}
A family of sets $\mathcal{S} \subset 2^{[n]}$ is an $(p, \varepsilon)$-robust clique sunflower (with core $C = \bigcap_{S \in \mathcal{S}} S$), if 
\begin{equation*}
    \Pr_G(\exists S \in \mathcal{S}, K_S \subset G \cup K_C) \geq 1-\varepsilon,
\end{equation*} 
where $G$ is a random \ErdosRenyi graph $\mathcal{G}_{n,p}$ sampled by including each edge independently with probability $p$, and $K_S \subset \binom{[n]}{2}$ is a clique on vertices $S$, i.e. $K_S := \{ \{u, v\} : u, v \in S, u\not=v\}$.
\end{definition}
This notion is intimately related to a similar notion of a robust sunflower, originally introduced in~\cite{rossman2010monotone}.
\begin{definition}[Robust sunflower]
\label{def:robust-sunflower}
A family of sets $\mathcal{S} \subset 2^{[n]}$ is an $(p, \varepsilon)$-robust sunflower (with core $C = \bigcap_{S \in \mathcal{S}} S$), if 
\begin{equation*}
    \Pr_W(\exists S \in \mathcal{S}, S \subset W \cup C) \geq 1-\varepsilon,
\end{equation*} 
where $W$ is a random subset of $[n]$ chosen by including each element independently with probability $p$.
\end{definition}
As it turns out, any large enough family of sets of size $\ell$ contain a robust sunflower or a robust clique sunflower. We introduce a notation for the dependence between the size of the family the parameters of this sunflower.
\begin{definition}
\label{def:robust-sunflower-function}
    We define $\RB(\ell, p, \varepsilon)$ to be the smallest number such that any $\ell$-uniform set system of size at least $\RB(\ell, p, \varepsilon)$ contains a $(p, \varepsilon)$-robust sunflower.

    Similarly, we define $\RCB(\ell, p, \varepsilon)$ to be the smallest number such that any $\ell$-uniform set system of size at least $\RCB(\ell, p, \varepsilon)$ contains a $(p, \varepsilon)$-robust clique-sunflower.
\end{definition}

Note that a family $\mathcal{S} \subset 2^{[n]}$ is a robust clique sunflower (as in the \Cref{def:robust-clique-sunflower}), if and only if the family of edge-sets $\{ K_S : S \in \mathcal{S}\}$ is a robust sunflower.
This implies a simple (yet unsatisfactory) bound
\begin{equation}
\label{eq:simple-rcb-rb-comparison}
\RCB(\ell, p, \varepsilon) \leq \RB\left(\binom{\ell}{2}, p, \varepsilon\right).
\end{equation}

In~\Cref{sec:intro-clique-sunflower-bounds} (for instance~\Cref{thm:clique-bound-rcb}) we provide a quantitative statement for how the upper bounds on $\RCB$ can be translated into lower bounds for the $\cliquep$ problem. 

For completeness, let us discuss a way to reinterpret a result similar to the Alon-Boppana (with the negative distribution being $\mathcal{G}_{n,p}$ instead of random $(\alpha-1)$-partite graph) in this framework, by deducing a bound on the robust clique sunflowers from the sunflower bound.

\begin{definition}[$k$-Sunflower]
    An family $\mathcal{S} \subset 2^{[n]}$ of $k$ sets is a $k$-sunflower (with a core $C = \cap_{S \in \mathcal{S}} S$), if all pairwise intersections of sets from $\mathcal{S}$ are $C$, i.e. for all $S_1 \not= S_2 \in \mathcal{S}$, we have $S_1 \cap S_2 = C$.
\end{definition}

What we call a $k$-sunflower is sometimes called in the literature \emph{a sunflower with $k$-petals}. We chose a slightly more concise terminology.

A sunflower lemma~\cite{erdos60} now says that for every $\ell$ and $k$ there is a finite number $S(\ell, k) \leq \ell! (k-1)^\ell$ such that every $\ell$-uniform\footnote{$\ell$-uniform set system is just a family of subsets of a universe, each subset of size exactly~$\ell$.} set system of size at least $S(\ell, k)$ has a $k$-sunflower. Subsequent results, including the breakthrough by~\cite{alweiss20} and improvements~\cite{bell2021note,Rao2020Coding,Tao2020sunflower, Hu2021sunflower, rao2023sunflowers} provide successively smaller upper bounds on $S(\ell, k)$, concluding with the best currently known bound
\begin{equation}
\label{eq:new-sunflower-numbers}
    S(\ell, k) \leq O(k \log \ell)^{\ell}.
\end{equation}
while the famous sunflower conjecture (\cite{erdos60}) stipulates that $S(\ell, k) \leq O(k)^\ell$.

The following observation, similar in form to the core of the Razborov and Alon-Boppana results, is a simple way to connect sunflowers to robust clique sunflowers.
\begin{lemma}
\label{lem:sunflower-is-clique-sunflower}
    Every $\ell$-uniform $k$-sunflower is an $(p, \exp(-k p^{\binom{\ell}{2}}))$-robust clique sunflower. In particular
    $\RCB(\ell, p, \varepsilon) \leq S(\ell, \log(1/\varepsilon) p^{-\binom{\ell}{2}})$.
\end{lemma}
\begin{proof}
    Consider a $k$-sunflower $S_1, \ldots S_k$ with core $C$. By the definition of sunflower the sets of edges $K_{S_i} \setminus K_{C}$ are disjoint, hence the indicator random variables
    $R_i := \mathbf{1}[K_{S_i} \subset G \cup K_C]$ are independent. Since $\E R_i \geq p^{\binom{\ell}{2}}$, the probability that all $R_i$ are zero is at most $(1 - p^{\binom{\ell}{2}})^k \leq \exp(-k p^{\binom{\ell}{2}})$.
\end{proof}

This, combined only with the classical upper bound $S(\ell, k) \leq O(\ell k)^{\ell}$ by \ErdosRado yields an upper bound
\begin{equation*}
    \RCB(\ell, p, \varepsilon) \leq O(1/p)^{\ell^3} (\ell \log(1/\varepsilon))^\ell,
\end{equation*}
which can be used to recover $n^{\Omega(\sqrt{\alpha})}$ lower bound via the approximation method outlined in the previous section~(see for instance~\Cref{thm:clique-bound-rcb}).

Crucially for us, a bound on robust sunflower numbers is an essential part of the new, improved series of results on the bounds for the sunflower numbers. Specifically, the following theorem was used to deduce those new bounds, and is important on its own for our application.

\begin{theorem}[\cite{alweiss20,bell2021note,rao2023sunflowers}]
\label{thm:robust-sunflowers}
    The robust sunflower numbers are bounded as
    \begin{equation*}
        \RB(\ell, p, \varepsilon) \leq O(p^{-1} (\log(1/\varepsilon) + \log \ell) )^{\ell}.
    \end{equation*}
\end{theorem}
This theorem fairly easily implies the breakthrough upper bound on sunflower numbers~\eqref{eq:new-sunflower-numbers}. For our purposes, though, the sunflower consequence of~\Cref{thm:robust-sunflowers} is less relevant, and we will leverage the existance of robust sunflowers more directly.

In~\Cref{sec:reduction} we prove a much stronger reduction than~\eqref{eq:simple-rcb-rb-comparison} showing how robust sunflower bounds can be used to give bounds on robust clique-sunflowers, that leads to our main improvement.

\begin{restatable}{theorem}{thmrcbcomparison}
\label{thm:rcb-comparison}
For any $\ell \geq 1$, $p \in (0,1)$ and $\varepsilon \in (0,1)$ we have
\begin{equation*}
    \RCB(\ell, p, \varepsilon) \leq \RB(\ell, p^\ell, \varepsilon/ \ell^2),
\end{equation*}
where $\RCB(\ell, p, \varepsilon)$ and $\RB(\ell, p, \varepsilon)$ are defined as in \Cref{def:robust-sunflower-function}.
\end{restatable}

Combining the upper bound for robust sunflowers in~\Cref{thm:robust-sunflowers} with our comparison theorem~(\Cref{thm:rcb-comparison}), we get
\begin{corollary}
\label{cor:our-rcb}
    The clique sunflower numbers are bounded as
    \begin{equation*}
        \RCB(\ell, p, \varepsilon) \leq O(p^{-\ell} (\log(1/\varepsilon) + \log \ell))^{\ell}.
    \end{equation*}
\end{corollary}

This is an improvement over a bound from~\cite[Lemma 3.2]{cavalar2022} of the same quantity. They showed an upper bound
\begin{equation}
\label{eq:cavalar-rcb-bound}
    \RCB(\ell, p, \varepsilon) \leq p^{-\binom{\ell}{2}} O(\ell \log (1/\varepsilon))^\ell.
\end{equation}
As we will see later, reducing the $O(\ell)^\ell$ factor down to $O(\log \ell)^\ell$ allowed us to show a lower bound $2^{\tilde\Omega(\sqrt{n})}$ for the clique problem as opposed to $2^{\tilde\Omega(n^{1/3})}$ from~\cite{alon87,cavalar2022}.

In order to prove~\Cref{thm:rcb-comparison} we show that any $(p^\ell, \varepsilon/\ell^2)$-robust sunflower is an $(p, \varepsilon)$-robust clique sunflower. This statement is easier to interpret for sunflowers with empty core: given a $\ell$ uniform family of sets $\mathcal{F} \subset 2^{[n]}$ (with empty common intersection), if we are very likely to cover one of those sets while sampling vertices of $[n]$ independently at random with probability $p^{\ell}$, we are also very likely to cover one of those with a clique, when sampling edges of $\binom{[n]}{2}$ independently with probability $p$.

As it turns out, one can set up specific stochastic processes $\{Y_S\}_{S \in \mathcal{F}}$, and $\{Y'_S\}_{S \in \mathcal{F}}$ associated with both of those experiments~--- the expected supremum of these processes will be directly connected (respectively) with the probability of covering one of the sets of $S \in \mathcal{F}$ by a random set (with the inclusion probability $p^\ell$), or the probability of covering one of the cliques on the vertices of $S \in \mathcal{F}$ by a random graph (with edge probability $p$). The main technical part of the proof is the comparison lemma stating that $\E \sup_{S \in \mathcal{F}} Y_S \leq \E \sup_{S \in \mathcal{F}} Y'_S$~---theorems of that form appeared in the literature on the theory of stochastic processes (most well known is the Slepian lemma, or Gaussian comparison principle~\cite[Corollary 2.10.12]{talagrand2014upper}). Even though we could not apply any of the known comaprison lemmas black-box, we adapt the ideas that were in~\cite{talagrand93} to show a comparison lemma for coordinate-wise contractions of cannonical Bernoulli processes, and we show the desired inequality for our two specific stochastic processes in question.
\subsection{Lower bounds for monotone circuits depending on sunflower bounds \label{sec:intro-clique-sunflower-bounds}}
We carry over the high-level structure of the proof outlined in~\Cref{sec:intro-approx-method} in more details in~\Cref{sec:approximation-method}. This part of the proof is technically very similar to known results applying the approximation method (for example \cite{razborov1985lower,alon87,cavalar2022}), although in contrast with the previous works we made an effort to provide a statement of the lower bound theorem parameterized by the robust clique sunflower numbers $\RCB(\ell, p, \varepsilon)$ --- which, in turn, can be upper bounded in several ways by the robust sunflower numbers $\RB(\ell, p, \varepsilon)$, or sunflower numbers $S(\ell ,k)$ as discussed in \Cref{sec:intro-sunflowers}.

Making these dependencies explicit, as opposed to choosing optimal values of relevant parameters ahead of time based on currently best known bounds on, say, sunflower numbers, makes the interplay between sunflower-like combinatorial statements and monotone lower bounds for the clique problem much clearer.

Moreover, our analysis allows us to show the inapproximability results (\Cref{thm:main-lb}) --- i.e. hardness of distinguishing a random clique of size $\beta$, from a random graph $\mathcal{G}_{n,p}$ which is unlikely to have even a clique of size $\alpha$ for $\alpha < \beta$. 

The analysis of \cite{razborov1985lower,cavalar2022} focused on the excact case, i.e. $\beta = \alpha$; whereas in \cite{razborov1985lower,alon87}, the negative distribution was chosen to be a random complete $(\beta-1)$-partite graph, but this choice was not crucial in their reasoning~--- only simple modifications are needed to adapt their proofs to the $\mathcal{G}_{n,p}$ being the negative distribution.

Specifically, we prove the following.
\begin{restatable}{theorem}{cliqueboundrcb}
\label{thm:clique-bound-rcb}
    Fix some $c \leq n$.  If for every $\ell \leq 2c$, we have
    $$(\beta / n) \RCB(\ell, p, \varepsilon)^{1/\ell} \leq \gamma \leq o(1)$$
    then the size of any monotone circuit $C$ distinguishing $\mathcal{T}^-_\alpha$ and $\mathcal{T}^+_\beta$ satisfies
    $$size(C) \geq \Omega(\min\{\gamma^{-c}, \ n^{-2c}/\varepsilon \}).$$
    In particular, choosing $\varepsilon = n^{-4c}$, if for all $\ell \leq 2 c$ we have
    \begin{equation*}
        (\beta / n) \RCB(\ell, p, n^{-4 c})^{1/\ell} < n^{-\delta}
    \end{equation*}
    then the monotone complexity of distinguishing $\mathcal{T}^-_\alpha$ and $\mathcal{T}^+_\beta$ is  $\Omega(n^{\delta c})$.
\end{restatable}

It is instructive to see how to essentially recover Alon-Boppana bounds from~\Cref{thm:clique-bound-rcb}, using a simple lemma stating that large enough sunflowers are robust clique sunflower (\Cref{lem:sunflower-is-clique-sunflower}) together with the new bounds on the sunflower numbers \eqref{eq:new-sunflower-numbers}. In this case we have an upper bound on the \begin{equation*}
\RCB(\ell, p, \varepsilon)^{1/\ell} \leq p^{-\ell^2} (\log(1/\varepsilon) + \log \ell) \leq n^{O(\frac{c^2}{\alpha})} c \log n,
\end{equation*}
so taking $c := O(\sqrt{\delta \alpha})$ for a small constant $\delta$, such that $\RCB(\ell, p, \varepsilon)^{1/\ell} \leq n^{\delta} \sqrt{\alpha} \log n$, we end up with a complexity lower bound $n^{\Omega(\delta c)} = n^{\Omega(\sqrt{\alpha})}$, as long as
\begin{equation*}
    \sqrt{\alpha} \beta \log n \lesssim n^{1 - 2\delta},
\end{equation*}
recovering \cite[Theorem 3.11]{alon87}. Setting $\alpha = \beta$, gives a lower bound of form $n^{\Omega(\sqrt{\alpha})}$ as long as $\alpha \leq n^{2/3 - \delta}$.

In a similar vein, plugging in the upper bound~\eqref{eq:cavalar-rcb-bound} for the $\RCB(\ell, p, \varepsilon)$ (originally shown in \cite[Lemma 3.2]{cavalar2022}), yields $n^{\Omega(\alpha)}$ lower bound, as long as $\alpha^2 \beta \leq n^{1-\delta}/\log n$, and again chosing $\alpha=\beta$ yields an $n^{\Omega(\alpha)}$ for $\alpha \leq n^{1/3 - \delta}$, recovering their \cite[Theorem 3.23]{cavalar2022}.

Finally,~\Cref{thm:clique-bound-rcb} together with our new bounds for the robust clique numbers~(\Cref{cor:our-rcb}) directly imply the claimed lower bound in~\Cref{thm:main-lb}.

\section{Approximating Cliques \label{sec:approximation-method}}

In this chapter, we will give an extended version of the approximation method, and leverage our improved robust clique sunflower bound to strengthen the lower bound for $\cliquep_k$ to $n^{\Omega(k)}$ if $k\leq n^{1/2 -\delta }$.

\subsection{Abstract Approximation Method}



For sake of abstraction, we describe a general inductive procedure of converting any monotone circuit $C$ computing a distributional decision problem $\mathcal{D}(\mathcal{T}^-, \mathcal{T}^+)$, gate-by-gate, into an approximation circuit $\hat{C} \in \mathcal{A}$ from a set of approximation circuits $\mathcal{A}$. This procedure depends on a pair of "compression" function $\mathcal{P}^{\land}, \mathcal{P}^{\lor}: \mathcal{A} \times \mathcal{A} \rightarrow \mathcal{A}$. The error introduced by the conversion, with respect to the distributions $\mathcal{T}^+$ and $\mathcal{T^-}$, will solely depend on $\mathcal{P}$.


\begin{definition}
\label{def:approximation}
    For any monotone circuit $C$ and a pair of compression  functions $\mathcal{P} = (\mathcal{P}^\wedge, \mathcal{P}^\vee)$ with $\mathcal{P}^\wedge, \mathcal{P}^\vee: \mathcal{A} \times \mathcal{A} \rightarrow \mathcal{A}$ define $\hat{C} \in \mathcal{A}$ as the result of the following procedure.
    \begin{enumerate}
        \item \textbf{Input variables}: Let $C=x_i$ be an input variable. Define 
        \begin{equation*}
            \hat{C} = C = x_i.
        \end{equation*}
        \item {$\wedge$ \textbf{-gate}}: Assume $C = \hat{C_0} \wedge \hat{C_1}$. Define
        \begin{equation*}
            \hat{C} = \mathcal{P}^\wedge(\hat{C}_0, \hat{C}_1).
        \end{equation*}
        \item {$\vee$ \textbf{-gate}}: Assume $C = \hat{C_0} \vee \hat{C_1}$. Define 
        \begin{equation*}
            \hat{C} = \mathcal{P}^\vee(\hat{C}_0, \hat{C}_1).
        \end{equation*}
    \end{enumerate}
\end{definition}

Depending on the choice of $\mathcal{P}$, $\hat{C}$ can be very different from $C$. For the sake of analysis, we are interested in the maximum error a single transformation step can introduce on either distribution.

\begin{definition}
\label{def:one-step-errors}
    Let $\mathcal{I}(\mathcal{P}')$ be the image of $\mathcal{P}'$ and define the positive approximation error
    \begin{equation*}
        \zeta^+_{\mathcal{P}} = \max_{\odot \in \{\wedge, \vee\}} \max_{\hat{C_0}, \hat{C_1} \in \mathcal{I}(\mathcal{P}^\odot)} \{ \Pr_{G \in \mathcal{T}^+}[(\hat{C_0} \odot \hat{C_1})=1 \text{ and } \mathcal{P}^{\odot}(\hat{C_0} \odot \hat{C_1})] = 0 \}
    \end{equation*}

    and the negative approximation error 
    \begin{equation*}
        \zeta^-_{\mathcal{P}} = \max_{\odot \in \{\wedge, \vee\}} \max_{\hat{C_0}, \hat{C_1} \in \mathcal{I}(\mathcal{P}^\odot)} \{ \Pr_{G \in \mathcal{T}^-}[(\hat{C_0} \odot \hat{C_1})=0 \text{ and } \mathcal{P}^\odot(\hat{C_0} \odot \hat{C_1})] = 1 \}
    \end{equation*}

\end{definition}


We can also formalize the intuition, laid out earlier, about proving circuit size lower bounds using the relative and absolute approximation errors.


\begin{definition}
    We say that a circuit $C$ distinguishes distributions $\mathcal{T}^+$ and $\mathcal{T}^-$ if $\E_{X \sim \mathcal{T}^+} [C(X)] - \E_{X\sim \mathcal{T}^{-}} [C(X)] \geq 2/3$.
\end{definition} 
We say that a circuit is $\mathcal{P}$-simple if it is in the image of the compression function $\mathcal{P}$. 
\begin{lemma}
\label{lem:extended-approximation}
    If every $\mathcal{P}$-simple circuit $\hat{C}$ satisfies $\E_{X \sim \mathcal{T}^+} [\hat{C}(X)] - \E_{X\sim \mathcal{T}^{-}} [\hat{C}(X)] \leq 1/3$, then every circuit $C$ distinguishing $\mathcal{T}^+$ and $\mathcal{T}^-$ has 
    \begin{equation*}
        \size(C) \geq \Omega(\min(1/\zeta^+_{\mathcal{P}}, 1/\zeta^-_{\mathcal{P}})).
    \end{equation*}
\end{lemma}
\begin{proof}
    The transformation of $C$ into $\hat{C}$ introduces at most $\zeta^+$ error on $\mathcal{T}^+$ and $\zeta^-$ error on $\mathcal{T}^-$ per gate. As $C$ distinguishes the distributions, but $\E_{X \sim \mathcal{T}^+} [\hat{C}(X)] - \E_{X\sim \mathcal{T}^{-}} [\hat{C}(X)] \leq 1/3$, the error of the $\mathcal{P}$-transformation has to match up the absolute error of the $\mathcal{P}$-simple circuit $\hat{C}$, such that the theorem follows.
\end{proof}

\subsection{Proving Clique Lower Bounds}

In the following, we will specialize the definition of the approximation circuits $\mathcal{A}$ and the compression functions $(\mathcal{P}^\wedge, \mathcal{P}^\vee)$ to prove the lower bounds on $\mathcal{D}(\mathcal{T}^-_\alpha, \mathcal{T}^+_\beta)$.

\begin{definition}
    Let $A \subseteq [n]$ be a subset of $n$ vertices. Let $\mathcal{K}_A$ be the clique indicator function on $A$, such that $\mathcal{K}_A(G)=1 \Leftrightarrow K_A \subset G$. For any family of subsets of $[n]$, $\{A_1, \ldots, A_m\}$, define the approximation circuit as 
    \begin{equation*}
        A = \bigvee_{i=1}^{m} \mathcal{K}_{A_i}
    \end{equation*}
    Let $\mathcal{A}$ be the set of all approximation circuits.
\end{definition}

We will often freely switch between the formal definition of an approximator as a boolean circuit $A = \bigvee_{i=1}^{m} \mathcal{K}_{A_i}$ and its representation as a set system over the universe of graph vertices $A = \{A_1, \dots, A_m\} \subseteq 2^{2^{[n]}}$.

\begin{definition}
    \label{def:approximation-gates}
    For an "inner-compression" function $\icomp: \mathcal{A} \to \mathcal{A}$, define $\mathcal{P} = (\mathcal{P}^\wedge, \mathcal{P}^\vee)$ with
    \begin{equation*}
        \mathcal{P}^\wedge(\bigvee_i^u \mathcal{K}_{X_i}, \bigvee_j^v \mathcal{K}_{Y_i}) = \icomp(\bigvee_i^u \bigvee_j^v \mathcal{K}_{X_i \cup Y_j})
    \end{equation*}

    \begin{equation*}
        \mathcal{P}^\vee(\bigvee_i^u \mathcal{K}_{X_i}, \bigvee_j^v \mathcal{K}_{Y_i}) = \icomp(\bigvee_i^u \mathcal{K}_{X_i} \vee \bigvee_j^v \mathcal{K}_{Y_i})
    \end{equation*}
\end{definition}

We observe that for the identity $\icomp = id$, $\mathcal{P}$ introduces no error on the graph distributions $\mathcal{T}^-_\alpha$ and $\mathcal{T}^+_\beta$. The error thus only depends on the choice of $\icomp$. For constructing such $\icomp$, we will use the robust clique sunflowers.

\subsection{Robust Closures}
For a DNF formula $A := \bigvee_{i=1}^m \mathcal{K}_{A_i}$ we define $\cl_{p, \varepsilon}(A)$ to be a formula obtained from $A$ by repeatedly replacing any $(p, \varepsilon)$-robust clique sunflower $A_{i_1}, \ldots A_{i_k}$ by its core $C := \bigcap_{j} A_{i_j}$, as long as there is such a robust clique sunflower, interleaved with removing all sets $A_j$, s.t. $A_i \subset A_j$ for some $A_i$. The $\trim_c(A)$ will be a circuit obtained from $A$ by removing all sets $|A_{i}| > c$.

Finally, we will chose the inner compression function (used in~\Cref{def:approximation-gates}) as 
\begin{equation*}
    \tau(A) := \trim_c(\cl_{p, \varepsilon}(A)).
\end{equation*}

We say that a circuit $A$ is $(p, \varepsilon)$-closed, if there are no $(p, \varepsilon)$-robust sunflowrers among the sets $A_1, \ldots A_m$.
The bound on the number of sets $A_j$ of a given size for an $(p, \varepsilon)$-closed circuit is a direct consequence of the definition of a closed circuit, and the $\RCB$ numbers. 
\begin{fact}
\label{lem:minterms-bound}
    The number $\mathcal{M}_l(A)$ of clique indicators $A_i$  of size $l$ of in a $(p, \epsilon)$-closed circuit $A$ is bounded by $$\mathcal{M}_l(A) < \RCB(\ell, p, \varepsilon).$$
\end{fact}

\subsection{The lower bound}
A crucial property that we will exploit both in the analysis of the approximator as well as in the construction of an explicit algorithm solving $\mathcal{D}(\mathcal{T}^-_\alpha, \mathcal{T}^+_\beta)$ is that the positive test distribution $\mathcal{T}^+_{\beta}$ is "well-spread" in a sense that it is unlikely for a large clique-indicator to accept many instances of $\mathcal{T}^+_{\beta}$.

\begin{lemma}
\label{lem:minterm-accepts}
    For a clique $K_B$ with $|B| = \ell$, $$\mathbb{E}_{G \sim \mathcal{T}^+_{\beta}}[\mathcal{K}_B(G)] < (\beta / n)^\ell.$$
\end{lemma}
\begin{proof}
    A graph $G$, sampled from $G \sim \mathcal{T}^+_{\beta}$ is an isolated $\beta$-clique $A$ on $n$ vertices. Thus $K_B(G) = 1$ if and only if $B \leq A$ such that 
    $$Pr_{G \sim \mathcal{T}^+_{\beta}}[\mathcal{K}_B(G) = 1] = Pr_{K_A \sim \binom{n}{\beta}}[B \subseteq A] = \frac{\binom{n-\beta}{\beta - \ell}}{\binom{n}{\beta}} \leq (\beta / n)^\ell.$$
\end{proof}

This, together with a bound on the number of clique indicators of each size~(\Cref{lem:minterms-bound}) allows us to easily show that any $\mathcal{P}$-simple circuit cannot distinguish $\mathcal{T}^+_{\alpha}$ from $\mathcal{T}^-_{\beta}$~---either circuit like that trivially accepts every output (and hence makes a large error on the negative distribution), or it accepts only small fraction of inputs from the positive distribution.

\begin{lemma}
\label{lem:simple-circuits-fails}
    If $\hat{C} \neq 1$ is an approximator obtained by the compression function $\mathcal{P}$, then 
    $$\mathbb{E}_{G \sim \mathcal{T}^+_\beta}[\hat{C}(G)] \leq \sum_{l=2}^c (\beta / n)^l \RCB(l, \epsilon, p).$$
\end{lemma}

\begin{proof}
    The proof is essentially a simple application of a union bound. If $G \sim \mathcal{T}^+_\beta$ and $\hat{C}(G) = 1$, then there must exist some term $\mathcal{K}_{A_i}$ of $\hat{C}$ with $K_{A_i} \subset G$. By~\Cref{lem:minterms-bound} there are at most $\mathcal{M}_l(f)$ terms of size $l$ (for all $l \leq c$) and by \Cref{lem:minterm-accepts} for every term of size $l$ the probability that it accepts a positive instance is at most $(\beta/n)^l$. Thus
    $$Pr_{G \sim \mathcal{T}^+_\beta}[\hat{C}(G) = 1] \leq \sum_{l=2}^c (\beta / n)^l \mathcal{M}_l(f) \leq \sum_{l=2}^c (\beta / n)^l \RCB(l, \epsilon, p).$$
\end{proof}

\subsubsection*{Single-step Approximation Errors}

We proceed by giving bounds on the single-step approximation errors $\zeta^+$ and $\zeta^-$, induced by~$\mathcal{P}$.


\begin{lemma}
    \label{lem:positive-error}
    The single-step error introduced on the positive distribution is bounded as follows
    $$\zeta^+_{\mathcal{P}} \leq \sum_{l=c}^{2c} (\beta /n)^l \RCB(l, \epsilon, p).$$
\end{lemma}
 \begin{proof}
    If we have for some $G \sim \mathcal{T}^+_\beta$ that $C(G) = 1$ but $\tau(C)(G) = 0$, then there was some term $\mathcal{K}_{A_i}$ of $\cl_{p,\varepsilon}(C)$ of size larger than $c$ with $\mathcal{K}_A(G) = 1$, which was discarded during the trimming process. 

    We can union bound the probability of this happening, in a similar way as~\Cref{lem:simple-circuits-fails}, using the upper bound $M_\ell(\cl_{p, \varepsilon}(C)) \leq \RCB(\ell, p, \varepsilon)$ (\Cref{lem:minterms-bound}) on the number of terms of a given size $\ell$, and the upper bound on the probability that a given term of size $\ell$ accepts the positive distribution~(\Cref{lem:minterm-accepts}).
\end{proof}

We also notice that for a fixed closure factor, the negative approximation error is bounded by the probability bound given by the robust clique-sunflower.

\begin{lemma}
    \label{lem:negative-error}
    The single-step error introduced on the negative distribution is bounded as 
    $$\zeta^-_{\mathcal{P}} \leq \epsilon n^{2c}.$$
\end{lemma}
\begin{proof}
    Note that while applying $\cl_{p,\varepsilon}(A)$ operation, each set $S \subset [n]$ of size at most $2c$ can be added as a core of some robust clique sunflower at most once, as the core is always a strict subset of the sunflower indicators. Hence, we will repeat the process of replacing a robust sunflower by its core at most $n^{2c}$ times, in each step introducing error at most $\varepsilon$ on the negative distribution (by the definition of robust clique sunflower).
\end{proof}

\subsubsection*{Complexity}

By the results of the last sections, we can proceed to prove a quantified condition on the existence of a monotone circuit lower complexity bound on $\mathcal{D}(\mathcal{T}^+_\beta, \mathcal{T}^-_\alpha)$, depending on the robust-clique sunflower bound $\RCB(l, \epsilon, p)$. The proof will follow by a simple application of~\Cref{lem:extended-approximation} together with bounds on the error introduced in each step of the process, established in the previous section. We recall the statement of the theorem from~\Cref{sec:intro-clique-sunflower-bounds}.

\cliqueboundrcb*

\begin{proof}
This statement follows directly from~\Cref{lem:extended-approximation}. First, note that by~\Cref{lem:simple-circuits-fails} any $\mathcal{P}$-simple circuit either accepts all negative instances, or accepts a positive instance with probability 
\begin{equation*}
    \sum_{2\leq \ell \leq c} (\beta/n)^{\ell} \RCB(\ell, p, \varepsilon) \leq \sum_{2 \leq \ell \leq c} \gamma^\ell \leq O(\gamma^2) \leq o(1).
\end{equation*}
hence a $\mathcal{P}$-simple circuit cannot distinguish it $\mathcal{T}^-_{\alpha}$ from $\mathcal{T}^+_{\beta}$.

The error $\zeta_{\mathcal{P}}^- \leq \varepsilon n^{2c}$ was shown as~\Cref{lem:negative-error}. For the $\zeta_{\mathcal{P}}^+$ we can use~\Cref{lem:positive-error}, to obtain a bound
\begin{equation*}
    \zeta^+_{\mathcal{P}} \leq \sum_{\ell = c+1}^{2c} (\beta/n)^\ell \RCB(\ell, p, \varepsilon) \leq \sum_{\ell=c+1}^{\infty}\gamma^\ell \leq O(\gamma^{c}).
\end{equation*}
With those two bounds, we can now directly apply~\Cref{lem:extended-approximation} to deduce the first part of the theorem. The second part follows by choosing $\varepsilon := n^{-4c}$,  $\gamma := n^{-\delta}$, and simple algebraic manipulations.
\end{proof}

We can now deduce~\Cref{thm:main-lb} directly from this more general statement and the new upper bound on the $\RCB$ (\Cref{cor:our-rcb}).
\begin{proof}[Proof of~\Cref{thm:main-lb}]
By \Cref{cor:our-rcb}, taking $\varepsilon = n^{-4c}$, for $\ell \leq 2c$ we have
\begin{equation*}
    \RCB(\ell, p, \varepsilon)^{1/\ell} \lesssim p^{-\ell} (\log (1/\varepsilon) + \log \ell) \lesssim n^{\frac{4 c}{\alpha - 1}} c \log n,
\end{equation*}
hence if we pick $c := \delta (\alpha-1)/8,$
we will have
\begin{equation*}
    \beta \RCB(\ell, p, \varepsilon)^{1/\ell} \lesssim \alpha \beta n^{\delta/2} \log n.
\end{equation*}
Now, when $\alpha \beta \leq n^{1-\delta} / \log n$, this yields
\begin{equation*}
    \beta \RCB(\ell, p, \varepsilon)^{1/\ell} \lesssim n^{1-\delta/2},
\end{equation*}
and finally, applying~\Cref{thm:clique-bound-rcb}, we get a lower bound of form
\begin{equation*}
    \Omega(n^{\delta c /2}) \geq \Omega(n^{\delta^2 \alpha / 16}) \geq n^{\Omega(\delta^2 \alpha)}.
\end{equation*}
    
\end{proof}

\section{Robust clique numbers are upper bounded by robust sunflower numbers\label{sec:reduction}}
In this section we will show a reduction, proving that any robust sunflower bound implies a corresponding robust-clique-sunflower bound with slightly different parameters. For our convenience, let us recall the statement of the theorem.

\thmrcbcomparison*

The main technical lemma towards \Cref{thm:rcb-comparison} is a comparison principle for a concrete pair of stochastic processes.  The proof of this result has been inspired by~\cite{talagrand93}. Before we state the comparison lemma, let us provide a necessary definition.
\begin{definition} 
\label{def:proper-lifting}
We say that $\phi : \binom{[n]}{k} \to \binom{[n] \times [n]}{k\ell}$ is a \emph{proper lifting} if for every $i \in S$, we have $|\phi(S) \cap (\{i\} \times [n])| = \ell$. 
\end{definition}
Note that since the image of any set under the proper lifting has exactly $k \ell$ elements, and it has exactly $\ell$ elements in each of the $k$ rows corresponding to elements of $S$, we have for every $i\not\in S$, that $\phi(S) \cap (\{i \} \times [n]) = \emptyset$. The canonical examples of proper liftings to have in mind are $\phi(S) := S\times S$ or $\phi(S) := S\times [\ell]$, although we will need a slightly more complicated one in the proof of~\Cref{thm:rcb-comparison}.

The following comparison lemma, shows that among the specific class of the stochastic processes associated with a lifting $\phi$, the smallest process is associated with the lifting $\phi(S) = S\times [\ell]$.
\begin{lemma}
\label{lem:comparison-lemma}
    Let $\mathcal{F}$ be a $k$-uniform family of subsets of $[n]$, and random variables $\{A_{i,j}\}_{i,j \in [n]}, \{\tilde{A}_{i,j}\}_{i\in [n], j \in [\ell]}$ be independent and identically distributed.

    If $\phi$ is any proper lifting (as in~\Cref{def:proper-lifting}), then
    \begin{equation}
    \label{eq:comparison-lemma}
        \E \sup_{S \in \mathcal{F}} \sum_{(i, j) \in S \times [\ell]} A_{i,j} \leq \E \sup_{S \in \mathcal{F}} \sum_{(i,j) \in \phi(S)} \tilde{A}_{i,j}.
    \end{equation}
\end{lemma}
We will prove this lemma later in \Cref{sec:proof-of-comparison-lemma}. Before we do, let us see how it implies~\Cref{thm:rcb-comparison}. More concretely we will show that every $(p^\ell, \varepsilon/\ell^2)$-robust sunflower is also an $(p, \varepsilon)$-robust clique sunflower, and the~\Cref{thm:rcb-comparison} will be a direct corollary.

It is easier to gain the intuition of how to prove a statement of this form from~\Cref{lem:comparison-lemma}, if we consider a special case of a robust sunflower $\mathcal{S}$ with an empty core $C$ --- in this case we would apply~\Cref{lem:comparison-lemma} with proper lifting $\phi(S) := S\times S$. The $i$-th row of a matrix $A$ has sum equal to $\ell$ independently with probability $p^\ell$, hence by the robust sunflower definition, with high probability there exists a square of form $S \times [\ell]$ with sum $\ell^2$ --- this implies that the expected supremum over sums over $S\times [\ell]$ as in the left-hand side of~\eqref{eq:comparison-lemma} is very close to $\ell^2$, and hence also expected supremum over the sums over $S\times S$ (by the aforementioned lemma). We can deduce that with high probability there is a square $S\times S$ in matrix $\tilde{A}$ with sum $\ell^2$ --- a statement corresponding to the fact that some clique $K_S$ is covered by a random $\mathcal{G}_{n,p}$ graph.

The handling of the more general case, where the core $C$ is not necessarily empty, is only slightly more technical.

\begin{lemma}
\label{lem:robust-sunflower-is-clique-sunflower}
    If $\ell$-uniform family $\mathcal{S} \subset 2^{[n]}$ is a $(p^\ell, \varepsilon/\ell^2)$-robust sunflower, it is also a $(p, \varepsilon)$-robust clique sunflower.
\end{lemma}

\begin{proof}
Consider a matrix $A \in \{0, 1\}^{n \times \ell}$, and $\tilde{A} \in \{0, 1\}^{n \times n}$, with independent entries, each equal to $1$ with probability $p$. We can treat the part of the matrix $\tilde{A}$  below the diagonal (i.e. $\tilde{A}_{i,j}$ for $j < i$), as determining the adjacency matrix for a random graph $G$ (sampled according to the Erdős–Rényi distribution $\mathcal{G}_{n,p}$).

Let us assume without loss of generality that $C = [\ell-k]$ for some $k$. We can look at a $k$-uniform family $\mathcal{S}' := \{S \setminus C : S \in \mathcal{S}\}$. Consider also the set $W$ of all indices such that the corresponding row in the matrix $A$ is all-one: $W := \{i : \sum_{j\leq \ell} A_{i,j} = \ell \}$. Note that each element in $W$ is included independently with probability $p^\ell$. Since $\mathcal{S}$ is a robust sunflower with the core $C$, directly by definition of robust sunflower with probability at least $1-\varepsilon$, there is a set $S \in \mathcal{S}'$ covered by $W$, i.e. 
\begin{equation*}
    \sum_{i,j \in S \times [\ell]} A_{ij} = \ell k.
\end{equation*}

This means that with probability $1-\varepsilon/\ell^2 \geq 1-\varepsilon/(\ell k)$, the value of
\begin{equation*}
    \sup_{S \in \mathcal{S}'} \sum_{i,j \in S \times [\ell]} A_{ij} = \ell k,
\end{equation*}
and therefore
\begin{equation*}
    \E \sup_{i,j \in S \times [\ell]} A_{ij}\geq (1-\varepsilon/\ell k)\ell l = \ell k - \varepsilon.
\end{equation*}

We will now apply~\Cref{lem:comparison-lemma}
 to a family $\mathcal{S}'$ together with a lifting function $\phi(S) := S \times ([n-k] \cup S)$ to deduce
\begin{equation}
\label{eq:two}
    \E \sup_{S \in \mathcal{S'}} \sum_{i,j \in \phi(S)} \tilde{A}_{ij}\geq  \ell k - \varepsilon.
\end{equation}
On the other hand, denoting by $1 - p_0$ the probability that there exists $S \in \mathcal{S'}$, s.t. all entries of $\tilde{A}_{ij}$ over $(i,j) \in \phi(S)$ are $1$ we have
\begin{equation}
\label{eq:three}
 \E \sup_{S \in \mathcal{S'}} \sum_{i,j \in \phi(S)} \tilde{A}_{ij} \leq (1 - p_0) \ell k + p_0 (\ell k - 1) = \ell k - p_0,    
\end{equation}
and combining \eqref{eq:two} and \eqref{eq:three}, we get $p_0 \leq \varepsilon$.

Finally, as discussed above, if we consider a graph $G$ on $[n]$, where the edge $\{i,j\} \in G$ for $j < i$ if $\tilde{A}_{i,j} = 1$, this graph has exactly the distribution of $\mathcal{G}_{n,p}$. Moreover, when there exists $S \in \mathcal{S}'$, s.t. all entries $\tilde{A}_{i,j} = 1$ for $i \in S, j \in S \cup C$, then in particular the clique $K_{S \cup C}$ is contained in the graph $G \cup K_C$. Indeed, all edges within $C$ are already contained in $K_C$, and $\tilde{A}_{i,j} = 1$ for $i \in S, j \in S \cup C$ implies that all edges between $S$ and $S \cup C$ are present in $G$.

Probability of this happening is $1-p_0 \geq 1-\varepsilon$, as desired, finishing the proof that $\mathcal{S}$ is indeed a robust-clique-sunflower.
\end{proof}
\begin{proof}[Proof of~\Cref{thm:rcb-comparison}]
    This theorem follows as a direct corollary of~\Cref{lem:robust-sunflower-is-clique-sunflower} and relevant definitions.

    Indeed, every $\ell$-uniform family of size at least $RB(p^\ell, \varepsilon/\ell^2)$ contains a $(p^\ell, \varepsilon/\ell^2)$-robust sunflower, which by~\Cref{lem:robust-sunflower-is-clique-sunflower} is  $(p, \varepsilon)$-robust clique sunflower.
\end{proof}

\subsection{Proof of~\texorpdfstring{\Cref{lem:comparison-lemma}}{Lemma 3.2} \label{sec:proof-of-comparison-lemma}}

We will prove this statement by induction. Specifically, let us define $B_t := [t] \times [n]$, and $\bar{B}_t := ([n] \setminus [t]) \times [n]$ (so that $B_t \cup \bar{B}_t$ forms a partition of $[n]\times[n]$ into the first $t$ rows, and the remaining $n-t$ rows).

Take $\phi_t : \binom{[n]}{k} \to \binom{[n]\times [n]}{\ell k}$ defined as $\phi_t(S) := ((S \times [\ell]) \cap B_t) \cup (\phi(S) \cap \bar{B}_t)$ --- i.e. on the first $t$ rows $\phi_t$ agrees with $S \times [\ell]$, and on the remaining $n-t$ rows it agrees with $\phi(S)$. Note that $\phi_t(S)$ has always $\ell k$ elements, and moreover $\phi_0(S) = S\times [l]$, and $\phi_{n}(S) =\phi(S)$. As such it is enough to show that for any $t$ we have
\begin{equation}
\label{eqn:induction-step}
        \E \sup_{S \in \mathcal{F}} \sum_{(i, j) \in \phi_{t}(S)} A_{i,j} \leq \E \sup_{S \in \mathcal{F}} \sum_{(i,j) \in \phi_{t+1}(S)} \tilde{A}_{i,j},
\end{equation}
where $\{A_{ij}\}$ and $\{\tilde{A}_{ij}\}$ are two collections of independent and identically distributed random variables (with the same distribution).

In order to prove this statement, we will consider a coupling where $A_{i,j} = \tilde{A}_{i,j}$ for all $i \not=t+1$, and let us condition on all those variables. 

For any given $S$, we can decompose:
\begin{equation*}
    \sum_{(i,j) \in \phi_{t}(S)} A_{i,j} = \sum_{(i,j)\in \phi_{t}(S), i\not={t + 1}} A_{i,j} + \sum_{j : (t+1, j) \in \phi_{t}(S)} A_{t+1, j}, 
\end{equation*}
and similarly for $\phi_t$. Note that when $i\not=t+1$ the pair $(i,j) \in \phi_t(S)$ if and only if $(i, j) \in \phi_{t+1}(S)$, so if we denote
\begin{equation*}
    a(S) := \sum_{(i,j) \in \phi_t(S) i\not=t} A_{i,j},
\end{equation*}
and by $p_0(S) := \{j : (t, j) \in \phi_{t}(S)\}$, and similarly $p_1(S) := \{ j : (t, j) \in \phi_{t+1}(S)\}$, we have
\begin{align*}
    \sum_{(i,j) \in \phi_t(S)} A_{i,j} = a(S) + \sum_{j \in p_0(S)} A_{t, j},
\end{align*}
and
\begin{align*}
    \sum_{(i,j) \in \phi_{t+1}(S)} A_{i,j} = a(S) + \sum_{j \in p_1(S)} A_{t, j}.
\end{align*}

Finally, notice that if $t \in S$, $|p_0(S)| = \ell$ and $p_1(S) = [\ell]$, whereas if $t \not\in S$, we have $p_0(S) = p_1(S) = \emptyset$. As such, to finish the induction, we only need to show the following claim.
\begin{claim}
    Consider a finite sequence of pairs $(S_1, a_1), (S_2, a_2), \ldots (S_m, a_m)$, where $S_i \subset [n]$, $a_i \in \bR$ and $|S_i|$ is either $\ell$ or $0$.

    Let $\mathcal{D}$ be a distribution over $\mathbb{R}$, and let $A_1, \ldots A_n$ and $\tilde{A}_1, \ldots \tilde{A}_n$ be two sequences of independent random variables distributed according to $\mathcal{D}$. Moreover, let $b_0 = \max \{ a_i : S_i = \emptyset \}$ and $b_1 = \max \{a_i : S_i \not= \emptyset\}$. Then
    \begin{equation*}
        \E \sup_{i \leq m}(a_i + \sum_{j \in S_i} A_j) \geq \E \max(b_1 + \sum_{i \in [\ell]} \tilde{A}_i, b_0).
    \end{equation*}
\end{claim}
\begin{proof}
If all sets $S_i$ are empty, the statement is trivial. Let $i_0$ be an index s.t. $S_{i_0} = \emptyset$ and $a_{i_0} = b_0$, and similarly let $i_1$ be an index such that $S_{i_1} \not= \emptyset$, and $a_{i_1} = b_1$. Consider any permutation $\pi : [n] \to [n]$, s.t. $\pi(S_{i_0}) = [\ell]$. We can extend the sequence $\tilde{A}_i$ to $n$ variables $\{\tilde{A}_{i}\}_{i \leq n}$, and consider a coupling between $A$ and $\tilde{A}$, given by $A_{i} := \tilde{A}_{\pi(i)}$. Clearly, under this coupling both marginal distributions of $\{A_i\}_{i \leq n}$ and $\{\tilde{A}_i\}_{i \leq \ell}$ are the right ones, all we need to show is that for any given realization of the joint process we have
\begin{equation*}
    \sup_{i\leq m} (a_i + \sum_{j \in S_i} A_j) \geq \max(b_1 + \sum_{i \leq l} \tilde{A}_i, b_0).
\end{equation*}
Indeed, by construction, we have
\begin{align*}
    \sup_{i\leq m} (a_i + \sum_{j \in S_i} A_j) & = \sup_{i \leq m} (a_i + \sum_{j \in S_i} \tilde{A}_{\pi(j)}) \\
    & \geq \sup_{i \in \{i_0, i_1\} } (a_i + \sum_{j \in S_i} \tilde{A}_{\pi(j)}) = \max(b_1 + \sum_{i \leq l} \tilde{A}_i, b_0).
\end{align*}
\end{proof}
This lemma completes the proof of the inequality~\eqref{eqn:induction-step}, and we can conclude the proof of~\Cref{lem:comparison-lemma}, by chaining the inequalities~\eqref{eqn:induction-step} across all $t$, since $\phi_0(S) = S \times [\ell]$ and $\phi_n(S) = \phi(S)$.
\section{The upper bound}

We now turn towards the upper monotone complexity bound for $\mathcal{D}(\mathcal{T}^-_\alpha, \mathcal{T}^+_\beta)$ where we will construct an explicit circuit $\mathcal{C}$ for distinguishing  $\mathcal{T}^+_\beta$ and $\mathcal{T}^-_\alpha$.

A useful subroutine that can be implemented using monotone boolean circuits is the ability to "sort" the input in polynomial size, as a sorting network on a boolean input can be easily used to construct a monotone boolean sorting circuit. \footnote{For this, note that a sorting network essentially is a sequence of binary operations on the input, swapping two inputs if they are out of order. A single swap of two booleans $x$, $y$ can be trivially implemented by an $AND = min\{x, y\}$ and an $OR = max\{x, y\}$ gate.} This also implies that they can implement the threshold function $T_\tau$ which accepts exactly if are at least $\tau$ input variables set to $1$, by just wiring the output to the $\tau$-th output of the sorting network.

\subsubsection*{Properties of $\mathcal{T}^+_\beta$ and $\mathcal{T}^-_\alpha$}


$\mathcal{T}^+_\beta$ and $\mathcal{T}^-_\alpha$ have an important distinctive property that we will use for constructing $\mathcal{C}_n$.
\begin{lemma}
\label{lem:positive-and-negative-prob}
    For a clique indicator $\mathcal{K}_A$ with $|A| = l$, the probability of clique inclusion follows
    $$Pr_{G \sim \mathcal{T}^+_\beta}[\mathcal{K}_A(G) = 1] = (\beta/n)^l$$
    and 
    $$Pr_{G \sim \mathcal{T}^-_\alpha}[\mathcal{K}_A(G) = 1] = p^{\binom{l}{2}} = (n^{-2/\alpha - 1})^{\binom{l}{2}}$$
\end{lemma}
\begin{proof}
The first equation has already been proven in \Cref{lem:minterm-accepts}. The second equation follows trivially from the definition of $\mathcal{G}_{n, p}$ --- each of the $\binom{l}{2}$ edges of $A$ is present in $G$ independently with probability $p$.
\end{proof}

The interesting property is that, for the right choice of parameters, the "clique-probability" $K_A$ is noticeably larger on $\mathcal{T}^+_\beta$ than on $\mathcal{T}^-_\alpha$. Taking for instance $p := \Pr_{G\sim\mathcal{T}_\alpha^-}(\mathcal{K}_A(G) = 1)$ and $q :=\Pr_{G\sim\mathcal{T}_\beta^+}(\mathcal{K}_A(G) = 1)$, if we are able to pick the size of the clique $l$ such that $q > 10p$, we will be able to construct a circuit of size $\Omega(1/q)$ that accepts $\mathcal{T}^+_\beta$ and rejects $\mathcal{T}^-_\alpha$, by looking at randomly placed $m = \Omega(1/q)$ clique indicators, accepting the input if at least $9 m p $ of those clique indicators are accepting on $G$.


\subsubsection*{An explicit circuit}
The probabilistic \footnote{ We construct an explicit probabilistic circuit here, however, the easy direction in the Yaos Principle \cite{yao1977probabilistic} also directly implies a deterministic upper bound for the distributional problem.} monotone circuit $\mathcal{C}_n$ for an $n$-boolean input is defined by the following construction. For parameters $m, l, \tau$
\begin{enumerate}
    \item Choose $m$ $l$-element subsets $A_1, \dots, A_m \subseteq [n]$ independent uniformly at random and connect the associated clique indicators $K_{A_1}, \dots, K_{A_m}$ to the input 
    \item Connect the output of the clique indicators to the threshold function $T_\tau^m$
    \item Connect the output of $T_\tau^m$ to the output of $\mathcal{C}_n$
\end{enumerate}

The size of the entire circuit like that is $O(m\log m)$. We would like to see how to chose $m$ and $l$ in order for the circuit to reject $\mathcal{T}^-_{\alpha}$ and accept $\mathcal{T}^+_{\beta}$ with probability $3/4$.

We will use the following well-known fact for the analysis of the above circuit construction.


\begin{fact}
\label{fact:distinguish}
Let $X_1, \ldots X_m$ be Bernoulli random variables with $\E X_i \leq p$ (not necessarily independent). Then $\E \sum X_i \leq pm$, and by Markov Inequality, $\Pr(\sum X_i > 4 pm) \leq 1/4$.

On the other hand, if $X'_1, \ldots X'_m$ are independent  Bernoulli random variables with $\E X'_i \geq 5p$, then  as soon as $m \gtrsim 1/p\delta$, Chebyshev inequality implies $\Pr(\sum X'_i \leq 4 pm) \leq O(\delta^2).$
\end{fact}

We will use $X_i$ to be a clique indicator on random $l$ vertices under the distribution $\mathcal{T}^-_{\alpha}$, and $X'_i$ to be a clique indicator on random $l$ vertices under the $\mathcal{T}^+_{\beta}$.

In what follows we will discuss how to choose $l$ in a way such that indeed $\E X_i' \geq 5 \E X_i$, and take $\delta$ to be a small constant to get a circuit distinguishing $\mathcal{T}^-_{\alpha}$ from $\mathcal{T}^+_\beta$. For the random variables $X_i$ (defined as the clique indicators on the negative distribution), we use only the Markov inequality, and we do not need to worry about correlations between them. 

Let us quickly discuss that indeed, on the positive distribution, the random variables $X_i'$ and $X_j'$ are indeed independent~--- that is the case, since conditioned on any specific realization $G$ of $\mathcal{T}^+_{\beta}$~--- i.e. $G$ being a clique on random $\beta$ vertices, the random variables $X_{i}'$ and $X_j'$ for $i\not=j$ are independent (since they are clique indicators on subsets of size $l$ chosen independently at random), and have $\E[X_i | G] = \E [X_i]$ --- the position of the original $\beta$ vertices in a clique does not affect the probability that a random $l$ vertices are the subset of those $\beta$ vertices.

The fact gives a bound on how many clique indicators we need for differentiating the graph distributions when the individual clique probabilities are sufficiently far from each other.

\subsubsection*{Circuit Analysis}

All we need to do now is to choose the parameter $l$ in such a way that $q \geq 5p$, where $p := \E_{G\sim \mathcal{T}^-{\alpha}} \mathcal{K}_A(G)$ and $q := \E_{G\sim \mathcal{T}^+{\beta}} \mathcal{K}_A(G)$. The resulting circuit size will be of order $O(1/q)$.

\begin{proof}[Proof of~\Cref{thm:main-ub}]

We consider circuit $C_n$ constructed as discussed above, with the threshold $\tau := 9 (\beta/n)^l$.  

According to~\Cref{lem:positive-and-negative-prob} and~\Cref{fact:distinguish} in order to make sure that the circuit indeed rejects $\mathcal{T}^-_{\alpha}$ and accepts $\mathcal{T}^+_{\beta}$ for given $\alpha \leq \beta$, we would like to pick the smallest $l$ such that
\begin{equation*}
    (\beta/n)^l \geq 5` (n^{-\frac{2}{\alpha-1}})^{l^2}
\end{equation*}
leading to a circuit of size $O((n^{\frac{2}{\alpha - 1}})^{l^2})$. After a sequence of simple algebraic manipulations, and taking
$$
\gamma := \frac{\alpha - 1}{2} \frac{\log(n/\beta)}{\log n}
$$
this condition is implied by
$$
l \geq \gamma + C/\gamma 
$$
for some universal constant $C$, since in this case we have $l^2 \geq \gamma^2 + 2C$. Taking $l := \gamma + C/\gamma$, and observing that by assumption we have
\begin{equation*}
    \frac{\log(n/\beta)}{\log n} \leq \delta
\end{equation*} 
we obtain the circuit size
\begin{equation*}
    \size(C_n) \leq O(n^{\frac{2l^2 }{\alpha - 1}}) \leq  O(n^{\frac{2\gamma^2}{\alpha-1}}) \leq O(n^{\alpha \delta^2/2}).
\end{equation*}
\end{proof}

\printbibliography
\end{document}